\documentclass[10pt]{IEEEtran}
\usepackage{bm}
\usepackage{epsfig}
\usepackage{graphicx}
\usepackage{mathtools}
\usepackage{cite}
\usepackage{amsmath,amsfonts,amssymb}
\usepackage{eurosym}
\usepackage{delarray}

\newcommand{\val}{\mathcal{V}}

\newcommand{\el}{\ell_{0,1}}

\usepackage{xcolor}

\newcommand{\ul}{\underline}
\newcommand{\mr}{\mathrm}
\ifCLASSINFOpdf
\else
\fi
\usepackage{enumerate}
\usepackage{mathtools}
\usepackage{amsthm}
\usepackage{graphicx,mathptm}
\newtheorem{definition}{\bf Definition}
\newtheorem{proposition}{\bf Proposition}

\usepackage{url}
\usepackage{float}
\usepackage{siunitx}

\begin{document}
%
\title{\Huge Coupled Charging-and-Driving Incentives Design for Electric Vehicles in Urban Networks}

\author{
    
Benoit Sohet, Yezekael Hayel, Olivier Beaude and Alban Jeandin

\thanks{B. Sohet and Y. Hayel are with LIA/CERI, Univ. of Avignon, 84911, Avignon, FRANCE (e-mail: \{benoit.sohet, yezekael.hayel\}@univ-avignon.fr).}

\thanks{B. Sohet, O. Beaude and A. Jeandin are with EDF R\&D, MIRE \& OSIRIS Dept, EDF Lab' Paris-Saclay, 91120, Palaiseau, FRANCE (e-mail: \{benoit.sohet, olivier.beaude, alban.jeandin\}@edf.fr).\vspace{0.1 cm}}

}



\maketitle

\begin{abstract}

Electric Vehicles (EV) impact urban networks both when driving (e.g., noise and pollution reduction) and charging. For the electrical grid, the flexibility of EV charging makes it a significant actor in ``Demand Response" mechanisms. Therefore, there is a need to design incentive mechanisms to foster customer engagement. A congestion game approach is adopted to evaluate the performance of such electrical transportation system with multiple classes of vehicles: EV and Gasoline Vehicles. Both temporal and energy operating costs are considered. The latter is nonseparable as it depends on the global charging need of all EV, which is scheduled in time by a centralized aggregator in function of nonflexible consumption at charging location. Thus, driving and charging decisions are coupled. An adaptation of Beckmann's method proves the existence of a Wardrop Equilibrium (WE) in the considered nonseparable congestion game; this WE is unique when the charging unit price is an increasing function of the global charging need. A condition on the nonflexible load is given to guarantee the monotonicity of this function. This condition is tested on real consumption data in France and in Texas, USA. Optimal tolls are used to control this electrical transportation system and then computed in order to minimize an environmental cost on a simple network topology.

\end{abstract}



%
\IEEEpeerreviewmaketitle


\section{Introduction}

\subsection{Motivation}


In 2015, the whole transport sector accounted for about a quarter (23\%) of global energy-related greenhouse gas emissions \cite{iea2017}.
Locally, road transport may undermine urban well-being because of traffic congestion, local air pollution and noise.
Electric vehicles (referred to as EV afterwards, with both battery and plug-in hybrid technologies included) seem to be an answer to both low carbon mobility (if associated with a low carbon electricity production mix) and urban well-being.
However, with around 8.3 millions EV expected in France by 2035 (see middle scenario of \cite{rte2018}), while the subsequent country electrical energy consumption in 2035 will be lower than the actual one, the charging need may reach 10\% in terms of power~\cite{rte2018}, which may lead to local grid constraints, e.g., transformers aging, power losses.
Currently, even if the penetration rate of EV is not really significant at the scale of a country, it can be already substantial locally\footnote{See e.g. the case of ``\^Ile-de-France", with more than 20 000 EV in circulation: \url{http://www.automobile-propre.com/dossiers/voitures-electriques/chiffres-vente-immatriculations-france/} (in French).}.

However, the \textit{flexibility} of EV charging -- in terms of compatibility with end users mobility needs and technical capabilities for load management -- makes it a significant actor in ``Demand Response" mechanisms~\cite{JACQUOT18} which is an emerging field in ``smart grids"\footnote{This is particularly true in comparison to other typical electrical tasks, like heating, cooking, lighting, for which there is no potential to ``smartly" schedule the associated electricity consumption profile.}. It consists in controlling consumption profile by, e.g., postponing usages in time, or reducing the level of power consumed, and with different objectives for the electrical system: local management of production-consumption balance, mitigating the impact on the electricity network \cite{beaude2016}, constituting a ``Virtual Power Plant" by aggregating flexible usages (see, e.g., \cite{vasirani2013} in the case of EV), etc. It is revolutionizing the traditional paradigm of the electricity system, where almost only generation units were flexible to ensure its effective operation.
In this context, taking into account charging strategies into everyday EV driving decisions will become an important issue in smart cities, particularly for urban networks \cite{brenna2012}. A first problem is the design of charging incentives (e.g., under the form of pricing) to share -- in space and time -- public EV charging stations (or EVCS) \cite{tan2017}.
Directly related to this day-by-day setting is the problem of charging infrastructure sizing, in terms of: (a) number of charging points; (b) location in space, and; (c) available levels of power (for each charging point and/or at the scale of a station, or network of stations) \cite{lam2014}.

To solve these problems, EV driving decisions must be also taken into account because of the driving-and-charging coupling.
This coupling is clearly observable during widespread holidays departures or notable sport events: the majority of driving EV need to charge at public EVCS, where there could be a significant waiting time and available power reduction (when allocated/shared between plugged EV) due to simultaneous power demands.
As an example of incentive mechanism, Tesla EVCS may adapt the charging prices in order to encourage EV to charge in empty EVCS rather than congested ones\footnote{\url{https://www.tesla.com/support/supercharging}.}.
Another example of this coupling worth mentioning comes from the French company CNR (\textit{Compagnie Nationale du Rh\^one}) and its ``Move In Pure" charging subscription: in order to guarantee that the source of electricity is renewable, charging is incentivized to be scheduled in some hours of the day and locations.
In a more futuristic vision, the EV charging-and-driving coupling can be transposed into a charging-by-driving one, with an inductive charging system (under the road) as suggested in \cite{wei2017}.
\subsection{Game theory design}

The model introduced in this work takes into account this coupling between EV driving and charging decisions to offer an accurate representation of EV behavior in order to test incentives aimed at, e.g., mitigating their impact on the electricity network.

In our game theory model, drivers choose their travel path depending on the costs they face (here travel duration and energy consumption), which depend themselves on the choices of the others (here through traffic congestion, and in the proposed charging problem).
This \textit{congestion game} is \textit{heterogeneous} because multiple vehicle classes are considered (EV and gasoline vehicles, or GV), with different costs, and \textit{nonatomic} as the large number of vehicles considered is approximated as infinite, and replaced by a continuous mass.

The charging unit price is the solution of a charging problem managed by an \textit{aggregator}, which schedules in time the charging operation of the whole EV fleet.
For simplification purposes, its objective is supposed to be aligned with the one of the \textit{Electricity Network Operator} (ENO): minimizing local electricity distribution costs (centralized optimization problem).
A use case of this framework is when all EV sign a contract whereby they experience cheaper electricity fares in exchange of delegating the control of their charging operations to the aggregator (see, e.g., \cite{han2010} for a presentation of such stakeholder in the case of EV aggregation).

To design incentives on top of this model, the driving and charging parts can be seen as the lower level of a bilevel framework illustrated in Fig.~\ref{fig:schema_coupled}.
At the upper level, a \textit{Transportation Network Operator} (TNO) and the ENO may impose tolls and electricity fares that induce a financial cost to drivers, thus influencing their driving and charging choices.
This aspect is illustrated in the numerical section.

\begin{figure}
\centering
\includegraphics[scale=0.32]{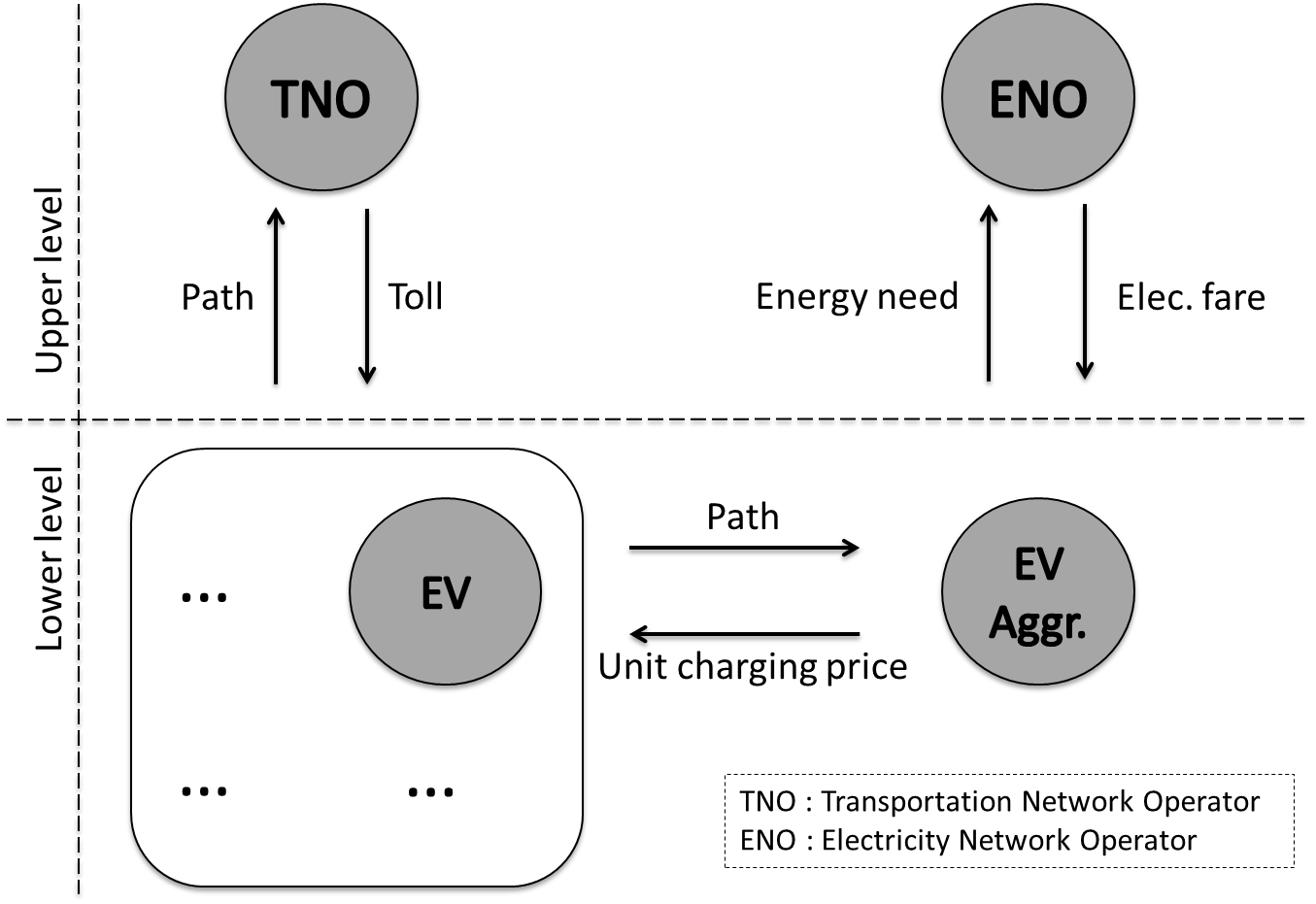}
\caption[Schematic representation of the coupled choice of driving and charging.]{Schematic representation of the coupled choice of driving and charging.
\textit{The driving problem gives an electricity charging need to the charging problem, which in turn provides a charging unit price. Our system is composed by four entities: Transportation Network Operator (TNO), Electricity Network Operator (ENO), Electric Vehicles (EV) and an EV aggregator.}
}
\label{fig:schema_coupled}
\end{figure}


Note that in this model, driving and charging decisions of EV users are coupled: the total energy charging need the aggregator schedules depends on the driving strategies of all EV.
In turn, the cost for a unit of energy used to drive -- named here \textit{charging unit price} -- is minimized by the aggregator depending on the total electricity distribution costs, which will impact the driving strategies.
Therefore, the charging unit price depends on the total energy charging need, which makes it \textit{nonseparable}: the energy consumption cost of an EV depends on the driving strategies of all EV, even those on different paths.
Finally, travel duration and energy consumption costs are nonlinear: the problem can be classified as a \textit{nonatomic multiclass congestion game with nonseparable and nonlinear cost functions}.
The interactions between the four different entities constituting this complex economical system are summarized here and in Fig.~\ref{fig:schema_coupled}:
\begin{itemize}
    \item (a large number of) Electric Vehicles (EV): each EV determines its travel path which minimizes its total cost, including congestion and energy,
        \item the EV aggregator: determines the optimal charging schedule for EV, minimizing the charging costs,
    \item the \textit{Transportation Network Operator} (TNO): determines optimal tolls on each road to control the traffic in the transportation network,
    \item the \textit{Electricity Network Operator} (ENO): determines the electricity fare in order to manage electricity demand and cover electricity production costs.
\end{itemize}

Our problem is related to intelligent electrical transportation systems in the sense that our model enables the control of the behavior of EV drivers. Game theory is a perfect tool to understand intelligent and adaptive behavior of drivers related to control. 

\subsection{Related works}

There is an emerging literature on the coupling of driving and charging decisions, identified in the review paper~\cite{WEI19}.
Let us detail the two contributions closest to our work in terms of methodology.
In~\cite{ALIZADEH17}, only EV are considered, and their charging choices are represented in an ``extended transportation network": each station is replaced by a set of virtual arcs, each one corresponding to a specific charging (energy) amount.
EV driving consumption is distance-dependent, each station has electricity production costs and sets a charging unit price maximizing its own profit.
This considered game is then separable and no theoretical result is given on the equilibrium. 
In~\cite{WEI18}, GV are added to the problem of~\cite{ALIZADEH17}, but EV energy consumption is not taken into account and the charging need is supposed to be the same for all EV.
Nevertheless, all stations belong to the same local grid so that a peak demand at one station impacts each station's charging unit price.
This makes their game nonseparable, though theoretical and numerical results are given on a (separable) decomposed model version.
Other works mentioned in~\cite{WEI19} are similar to the two previous ones, like~\cite{HE16} where the vehicle's destination offering a minimal cost is chosen, or \cite{ZHU18} where a fleet operator chooses the proportion of its vehicles to charge instead of taking customers.
Finally, in a more recent paper \cite{ALIZADEH19}, the authors propose a hierarchical model in which multiple charging network operators manage their own charging station, inducing a particular behavior of EV. The problem is formulated as a Mathematical Problem with Equilibrium Constraints (MPEC) which is NP hard. A heuristic based on solving a sequence of integer programs with convex objective and constraints is suggested to tackle the problem, but the computational complexity is still important. Numerical illustrations are proposed with a six nodes graph and a power network model associated with.

Unlike the previous references, our model does not rely on iterations between the driving and charging problems but offers a closed-form expression of the costs.
Our solution to deal with this added complexity is based on the following works.
Firstly, basic Traffic Assignment Problem (TAP) with single-class customers is defined and studied in \cite{SHEFFI85}. It is shown that when all customers are equally affected by the congestion (symmetric setting) and when the cost functions are increasing, the equilibrium is unique. In recent years, there has been an increasing interest for mixed TAP (MTAP) where two or more classes of vehicles are considered \cite{JIANG14}. 
Uniqueness of the Wardrop Equilibrium\footnote{The concept of equilibrium used in game theory in the case of a continuous mass of players.} (WE) in mixed TAP is proved in \cite{ALTMAN05} when the cost functions are the same for every customer, up to an additive constant. The complexity induced by mixing different types of traffic in a routing game comes from the difficulty to use standard approaches like Beckmann's one in order to determine the WE. In fact, different impacts on travel costs cause an asymmetry in the Jacobi matrix \cite{Dafermos72}. Therefore different techniques and approaches like nonlinear complementarity, variational inequality and fixed-point problems for characterizing equilibrium are reviewed in \cite{Forian95}.
Secondly, studies of nonatomic games with nonseparable costs are nontrivial and very few papers deal with this framework. In \cite{Chau03}, the authors generalize the bound obtained by Roughgarden and Tardos \cite{Rough2000} on the Price of Anarchy\footnote{Standard game theoretic measure of equilibria performance (w.r.t. a centralized optimum).} for nonseparable, symmetric and affine costs functions. In \cite{Perakis04}, the author considers a similar framework but with asymmetric and nonlinear costs. The bounds obtained are tight and are based on a semidefinite optimization problem. In \cite{Correa08}, the authors propose a new proof for the Price of Anarchy in nonatomic congestion games, particularly with nonseparable and nonlinear cost functions. Their geometric approach leads to obtain in a simple manner the bounds found in \cite{Chau03} and \cite{Perakis04}. Similarly to this game theoretic literature, we focus on the algorithmic part and techniques to characterize and compute an equilibrium based on potential functions and Beckmann's techniques, for nonatomic routing games with nonseparable cost functions. In fact, our framework induces particular cost functions which enables us to characterize the WE as the minimum of a global function. 

Our model was first presented in our previous work~\cite{SOHET18} on a toy example.
The main contributions of this paper are:
\begin{enumerate}
\item an explicit formulation of the charging unit price based on a water-filling approach as in~\cite{mohsenian10}, which depends on EV energy consumption and which is integrated in a coupled driving-charging model in a closed-form fashion. In particular, it extends the approach of \cite{ALIZADEH17} where the considered charging problem does not allow a temporal power scheduling, the charging cost being directly calculated based on the total EV energy need.
Furthermore, with the electrical grid costs internalized directly inside the charging unit price, there is no need in our work to iterate between the driving and charging problems to determine the game's equilibrium and therefore, no problem of convergence as in~\cite{ALIZADEH17}; 
\item the proof of existence of a WE in a nonseparable routing game -- by finding an equivalent Beckmann's function -- and its uniqueness under a ``smoothness condition" on the nonflexible load the ENO must meet (in addition to the EV charging need) in a general transportation network. Numerically, this condition is verified on real data sets and key sensitivity parameters are identified (seasonal effects and consumption habits);
\item the design of incentives to minimize an environmental cost for the TNO: numerically, optimal tolls can be computed by the proposed model; significant gains are obtained relatively to the no-toll case.\\
\end{enumerate}


\begin{table}[]
    \centering
    \begin{tabular}{cl}
     \hline
      $\mathcal{N}$   &  Set of nodes\\
     $\mathcal{K}$ & Set of Origin-Destination pairs \\
     $\mathcal{A}$    & Set of arcs (directed links)\\
     $\mathcal{R}$ & Set of paths (or routes)\\
     $D_k$ & Travel demand of O-D $k$\\
     $l_a$ & length of arc $a$\\
     $X_e$ & Proportion of EV \\
     $f_{r,s}$ & Flow rate of vehicle class $s$ on path $r$ \\
     $x_{a,s}$ & Flow rate of vehicle class $s$ on arc $a$\\
     $d_a(x_a)$ & Travel duration on arc $a$\\
     $d^0_a$ & Free flow reference time\\
     $C_a$ & Capacity of arc $a$ in vehicles per unit of time\\
     $\tau$ & Economic value of travel duration per unit of time\\
     $t_{a,s}$ & Toll for vehicle class $s$ on arc $a$\\
     $m_s$ & Energy consumption for vehicle class $s$ per distance unit\\
     $\lambda_s$ & Unit energy price for vehicle class $s$\\
     $c_{r,s}(\mathbf{x})$& Total driving cost for vehicle class $s$ on path $r$\\
     $T$ & Number of charging time slots \\
     $\ell_t$& Total electricity load during time slot $t$\\
     $\ell_{0,t}$& Nonflexible electricity consumption during time slot $t$\\
     $\ell_{e,t}$& Flexible electricity consumption during time slot $t$\\
     $f_t(\ell_t)$ & Local electricity distribution costs (i.e., electricity fare) during time slot $t$\\
     $L_e$& Total charging need\\
     $c_\text{env}(\mathbf{x})$ & Environmental cost\\
     \hline
     \vspace{0.2mm}
    \end{tabular}
    \caption{Table of Notation     \vspace{-6mm}
}
    \label{tab:my_label}
\end{table}

The paper is organized as follows.
Our coupled model is introduced in Sec.~\ref{sec:model} and the coupling between the driving and charging problems through the charging unit price is explained in Sec.~\ref{sec:charging}.
Equilibrium analysis of this model is proposed in Sec.~\ref{sec:optimal} and numerical illustrations based on real data analysis are exposed in Sec.~\ref{sec:numerical}. Finally, conclusions and perspectives are given in last Sec.~\ref{sec:conc}.
\section{A multiclass nonseparable routing game}
\label{sec:model}
Two types (or classes) of vehicles are considered: Electric Vehicles (EV) and Gasoline Vehicles (GV)\footnote{Note that the following study can be generalized to more classes.}.
All vehicle users determine rationally (i.e. optimally) their travel path, with perfect knowledge of all the costs. 
The novelty is that the costs considered here depend not only on the congestion of the driving path, but also on the energy used to travel along it. Moreover, the energy cost for EV depends on the global charging demand of electricity, which depends itself on the driving strategies of all EV in a way described in next section. \\


\ul{\textbf{Notations}}: bold symbols refer to vectors, subscript $e$ to EV variables and $g$ to GV ones. \\


\vspace{-2mm}

The road network is represented by a graph $G=(\mathcal{N},\mathcal{A})$, where $\mathcal{N}$ is the set of \textit{nodes} and $\mathcal{A}$ is the set of \textit{arcs} (directed links), each arc $a\in\mathcal{A}$ being associated with a length $l_a$ (expressed in km typically).
The set of \textit{paths} (or routes) of $G$ is $\mathcal{R}$, and a quantity $\delta_{a,r}$ is defined as equal to 1 if arc $a\in\mathcal{A}$ is part of the path $r\in\mathcal{R}$, 0 otherwise.
This way, the length of route $r\in\mathcal{R}$ can be expressed by: $l_{r}=\sum_a \delta_{a,r} l_a$.

The set of all \textit{Origin-Destination (O-D) pairs} is $\mathcal{K}$, and~each O-D pair $k\in\mathcal{K}$ is associated with a \textit{travel demand} $D_k$ and a set of different possible paths $r_k\in\mathcal{R}_{k}$ linking O to D (there may be more than one).
Travel demand $D_k$ represents the proportion of vehicles which need to travel from O to D per unit of time (e.g., per day, considering morning and evening commuting periods of working days) and is the input of the driving problem (set to $0$ if there is no demand for O-D pair $k$ and $\sum_k D_k = 1$).
Note that the driving problem is stationary: all demands $D_k$ are averaged over time.
The proportion of EV among all vehicles is written $X_e$, and the proportion of GV is then given by $X_g = 1-X_e$. 

Every path $r\in\mathcal{R}$ is associated with an EV, a GV and a total vehicle flow rates: respectively $f_{r,e}$, $f_{r,g}$ and $f_r = f_{r,e} + f_{r,g}$.
Meeting all travel demands imposes constraints on these flows:
\begin{equation}
\forall k\in\mathcal{K}\,,~ \forall s\in\{e,g\}\,,\quad
\begin{cases}
\forall r_k\in\mathcal{R}_k\,,~ f_{r_k, s} \geq 0\,,\\
\sum_{r_k\in\mathcal{R}_k} f_{r_k, s} = X_s ~D_k\,.
\end{cases}
\label{eq:contraintes}
\end{equation}
Here it is assumed that the travel demand on each O-D pair is proportionally distributed between classes: the demand for O-D pair $k\in\mathcal{K}$ by class $s\in\{e,g\}$ is $X_s~D_k$.

Similarly, every arc $a\in\mathcal{A}$ is associated with an EV, a GV and a total vehicle flow per unit of time: respectively $x_{a,e}$, $x_{a,g}$ and $x_a = x_{a,e} + x_{a,g}$, which are related to route flows as such:
\begin{equation}
\quad\quad\quad\quad\quad x_{a,s} = \sum_{k\in\mathcal{K}}\sum_{r_k\in\mathcal{R}_k} \delta_{a,r_k} f_{r_k,s}\, \quad (s=e,g) \, .
\label{eq:x_f}
\end{equation}
The vector containing all EV and GV arc flows is noted $\mathbf{x}$.\\
There are three sorts of costs on every arc $a\in\mathcal{A}$:

(i) The \textit{travel duration} on arc $a$ is given by the following congestion function (from the Bureau of Public Roads)~\cite{spiess90}:
\begin{equation}
d_a(x_a)= d^0_a\left[1+\alpha\left(x_a/C_a\right)^\beta\right],
\label{eq:BPR}
\end{equation}
where $C_a$ is the capacity of the link in vehicles per unit of time and $d_a^{0}=\frac{l_a}{v_{a}}$ is the \textit{free flow} reference time, with $v_a$ the maximum speed limit on the arc.
 The two remaining parameters $\alpha>0$ and $\beta>1$ are adjusted empirically.
 Note that travel duration is the same for both classes, and that it depends on the total flow on the arc $x_a$.
 
(ii) Travelers on arc $a$ are also subject to a \textit{toll}, $t_{a,e}$ or $t_{a,g}$, which depends on their vehicle class. This is the control of the TNO (see Fig.~\ref{fig:schema_coupled}).

(iii) Finally, vehicles are \textit{energy} consuming, and the driving cost should depend on this feature too. 
Here, the consumption model is distance-dependant: $m_e$ (resp. $m_g$), the electricity (resp. fuel, in liter) consumed per distance unit, is constant and does not depend on speed profiles, etc.
With $\lambda_s$ the charging/fueling unit price, the consumption cost on arc $a$ is ($s=e,g$):
\vspace{-0.24 cm}
\begin{equation}
\underbrace{l_a}_{\text{km}} \times \underbrace{m_s}_{\text{(kWh,L)/km}}\times \underbrace{\lambda_s}_{\text{\euro/(kWh,L)}} (\text{\euro})\,.
\label{eq:cons}
\end{equation}

Summing these three costs, the total driving~cost for a vehicle of type $s \in \{e,g\}$~choosing~route~$r\in\mathcal{R}$~is~given~by:
\begin{eqnarray}
\label{eq:gen_cost}
c_{r,s}(\mathbf{x}) = \sum\limits_{a\in\mathcal{A}} \delta_{a,r}\times\Big(\tau d_{a}(x_a)+ t_{a,s}+ l_a m_s \lambda_s\Big),
\end{eqnarray}
where $\tau$ is the cost of one unit of time spent driving. Note that here the price of fuel is assumed constant: $\lambda_g$ is an exogenous parameter, it does not depend on drivers decisions. On the contrary, the charging unit price $\lambda_e$ is endogenous as it will depend on EV driving strategies: this is the subject of the next section.
\vspace{-0.2 cm}
\section{Charging price: a centralized water-filling approach}
\label{sec:charging}

In our framework, the charging unit price, which is integrated into the driving cost function, is chosen by an \textit{aggregator} to balance its costs.
These costs are aligned with the ones of the ENO and limited to the global production cost: the losses associated with the dispatching of power to dispersed charging stations is not considered here and will be the subject of a future work.
Furthermore, this aggregator manages the charging operation of all EV by scheduling in time the total charging need in order to minimize its global production cost.
The resulting charging unit price then affects retroactively the path choices of EV in the stationary routing game through the driving cost function.


The scheduling only considers the aggregated (over all EV) energy need $L_e$: using \eqref{eq:cons} (constant energy consumption per distance unit), $L_e$ is proportional to the total travelled distance by all EV:
\vspace{-0.24 cm}
\begin{equation}
L_e = \sum_{k\in\mathcal{K}}\sum_{r_k\in\mathcal{R}_k}f_{r_k,e} l_{r_k} m_e
= m_e \sum_{a\in\mathcal{A}} x_{a,e} l_a \,.
\label{eq:L_e}
\end{equation}
Note that $L_e$ depends on all EV flows, and that a higher flow $x_{a,e}$ on an arc $a$ does not automatically lead to a higher $L_e$ because of the travel demand constraints~\eqref{eq:contraintes}: there might be fewer EV traveling on a longer arc $b$ at the same time.

The  scheduling problem is written in discrete time: the aggregator schedules the global charging need $L_e$ resulting from the stationary driving problem among a finite number $T\geq 2$ of time slots $t\in\{1,\dots,T\}$ (typically, this discretized period represents a day).
More precisely, the aggregator chooses which portion $\ell_{e,t}$ of the aggregated charging need $L_e$ to charge during each time slot $t$ in order to minimize its costs (here, the electricity generation).
This scheduling is not trivial since these costs may depend inherently on the time slot $t$ (e.g., solar panels produce only during the day) and are increasing with the total electricity load $\ell_t$ during that time slot $t$~\cite{WOOD12}. Here, the following ``proxy"\footnote{It does not include dynamic or location effects.} is considered to represent this cost:
\begin{equation}
 \forall t\in\{1,\dots,T\}\,, \quad   f_t(\ell_t) = \eta_t (\ell_t)^n\qquad\qquad(n\geq2)
\end{equation}
Note that quadratic cost functions $f_t$ ($n=2$) are good proxies, widely used in the literature~\cite{mohsenian10}: they will be considered in the numerical section.

At each time slot $t$, the total electricity load $\ell_t$ is made of two components: the portion $\ell_{e,t}$ of the aggregated charging need which is scheduled at $t$, and a \textit{nonflexible (fixed) consumption} $\ell_{0,t}$ which includes electrical usages that are present in charging locations where EV are plugged (e.g., household appliances when charging at home, tertiary ones at professional sites). While nonflexible consumption $\ell_{0,t}$ is a parameter of the charging problem, $\ell_{e,t}$ is the control variable. As modeled in \cite{mohsenian10} (and in many other papers about EV smart charging), EV consumption has to be scheduled depending on other electrical usages, the impact on the grid being dependent on the aggregate consumption (obtained as the sum of flexible and nonflexible profiles). Note that here, all EV share the same (local) electricity network. In an extended setting, the aggregate consumption profile could be calculated at different charging locations, with possibly different electricity network topologies, constraints and nonflexible profiles. 

 
%
Formally, the charging problem solved by the EV aggregator is stated as follows:
\vspace{-0.1 cm}
\begin{equation}
\label{eq:pb-flex-load-schedul}
\min_{\mathbf{\ell_e}=\left(\ell_{e,t}\right)_t} ~ \sum_{t=1}^T f_t\left(\ell_{0,t}+\ell_{e,t}\right), \quad \mbox{s.t.} \quad
\begin{cases}
\forall t\,,~~\ell_{e,t} \geq 0  \, ,\\
\sum_{t=1}^T\ell_{e,t}=L_e \, .
\end{cases}
\end{equation}
\noindent
For a global charging need $L_e$ issued from the driving problem (second constraint in \eqref{eq:pb-flex-load-schedul}), the aggregator has to determine~the aggregated charging profile $\bm{\ell}_e:=(\ell_{e,1},\ldots,\ell_{e,T})$ which minimizes the total charging cost. Note that this problem is parametrized by both the global charging need $L_e$ (from~the~driving~problem,~so~endogenous~in~the~coupled~charging-driving~model)~and~the~nonflexible~load~profile~$\bm{\ell}_0$~(exogeneous). Finally, note that the disaggregation of the aggregate charging profile $\bm{\ell}_{e}$ obtained as the output of this scheduling problem is assumed possible here: thanks to the large number of EV, the aggregated charging profile can be decomposed into charging profiles for every EV. In practice, the underlying assumption is that all EV can find an available charging station at any moment and any location (no reservation or queuing effect here) and that there are enough EV not driving for the charging purposes.

Suppose now without loss of generality (because there are no dynamical effect taken into account here) that time slots are ordered such that $f_1'(\ell_{0,1}) \leq \cdots \leq f_T'(\ell_{0,T}) $, adding time slot $T+1$ (with $\ell_{0,T+1} = + \infty$) to unify notations.
Considering the two following auxiliary parameters: 
\vspace{-0.1 cm}
\begin{equation}
\alpha_t = \sum_{s = 1}^t\ell_{0,s}\quad \mbox{and} \quad
\beta_t = \sum_{s = t+1}^{T} f_s(\ell_{0,s})\,,
\end{equation}
respectively representing the cumulative nonflexible load up to time $t$, and the cost due to nonflexible load after time $t$, the optimal value $\val$ of problem \eqref{eq:pb-flex-load-schedul} is given by the following Proposition.


\begin{proposition}
Given a nonflexible vector $\bm{\ell}_{0_{_{_{}}}}$, if the global charging need verifies $L_e \in \left]L_e^{(\bar{t}-1)},L_e^{(\bar{t})}\right]$ with the energy thresholds $L_e^{(t)}=  \left[\sum_{s=1}^t \left(\eta_{t+1}/\eta_s\right)^{1/(n-1)}\right]\ell_{0,t+1}-\alpha_t$ for ${t\in\{1,\dots,T\}}$,
then the solution of the energy scheduling problem~\eqref{eq:pb-flex-load-schedul} yields the optimal value:
\begin{equation}
\val(L_e)= (\sum_{s=1}^{\bar{t}} \eta_s^{-\frac{1}{n-1}})^{-(n-1)} \times \left(L_e+\alpha_{\bar{t}}\right)^n~+~\beta_{\bar{t}}\,.
\label{eq:water_filling_min}
\end{equation}
\label{prop:water_filling}
\end{proposition}
\vspace{-0.4 cm}

\begin{proof}
\label{append:water_filling}
Without loss of generality, time slots are assumed to be ordered by marginal cost $f'_t(\ell_{0,t})$, i.e. $f'_1(\ell_{0,1})\leq f'_2(\ell_{0,2})\leq \ldots \leq f'_T(\ell_{0,T})$.
At the optimal scheduling of a given charging need $L_e$, time slots used share the same marginal cost, lower than the marginal costs of the unused time slots.
This way, the aggregator schedules $L_e$ in the order of time slot indices ($t=1$, then $t=2$, etc.).

As a consequence, energy thresholds $L_e^{(t)}$ are defined as the charging need from which the aggregator starts using time slot $t+1$ ($0\leq t < T$).
At each threshold, the charging need $L_e^{(t)}$ is scheduled among the time slots already used ($s\leq t$) so that the resulting marginal costs are all equal to the one of the empty time slot $t+1$: $f'_s\left(\ell_{0,s}+\ell_{e,s}^{(t)}\right) = f'_{t+1}\left(\ell_{0,t+1}\right)$.
As cost functions $f_t$ are convex, this method ensures that the marginal cost associated with each infinitesimal portion of $L_e^{(t)}$ was lower than $f'_{t+1}\left(\ell_{0,t+1}\right)$.
Mathematically, $L_e^{(t)}=\sum_{s=1}^t \ell_{e,s}^{(t)}$ with: 
\vspace{-1 mm}
\begin{equation*}
\left\{
\begin{array}{ll}
\ell_{e,s}^{(t)} = \left(\frac{\eta_{t+1}}{\eta_s}\right)^{n-1}\ell_{0,t+1} -\ell_{0,s} & \text{for } s \leq t \,,\\
    \ell_{e,s}^{(t)} = 0 & \text{for } s>t\,,
\end{array}
\right.
\end{equation*}
which yields the energy thresholds formula.

Thus, if $\bar{t}$ is such that the charging need $L_e$ is in $]L_e^{(\bar{t}-1)}, L_e^{(\bar{t})}]$, then the optimal charging profile $\bm{\ell_e}$ is such that for $t>\bar{t}$, $\ell_{e,t}=0$  and for $t\leq\bar{t}$, $f'_t\left(\ell_{0,t}+\ell_{e,t}\right) = f'_1\left(\ell_{0,1}+\ell_{e,1}\right)$, so that:
\begin{equation*}
\begin{aligned}
\val(L_e) = \sum_{t=1}^{\bar{t}} \eta_t \left(\frac{\eta_1}{\eta_t}\right)^\frac{n}{n-1}\left(\ell_{0,1}+\ell_{e,1}\right)^n + \beta_{\bar{t}}\,.
\end{aligned}
\end{equation*}

Finally, $\ell_{e,1}$ is deduced using constraint $L_e = \sum_{t=1}^T\ell_{e,t}$\,:
\begin{equation*}
    \sum_{t=1}^{\bar{t}} \left(\frac{\eta_1}{\eta_t}\right)^\frac{1}{n-1} \left(\ell_{0,1}+\ell_{e,1}\right) = L_e +\alpha_{\bar{t}}\,.
\end{equation*}
\vspace{-5mm}

\end{proof}

Note that when electricity costs functions $f_t$ are not time-dependent ($\forall t\,,~\eta_t = \eta$ so that $f_t=f$), the optimal aggregated charging profile $\bm{\ell}_e$ of Prop.~\ref{prop:water_filling} has a water-filling structure \cite{mohsenian10}.

Having optimally scheduled the aggregated charging need $L_e$, the aggregator then determines the charging unit price:
\begin{eqnarray}
\lambda_e(L_e)
=\frac{\val(L_e)}{L_e+\sum_{t=1}^T \ell_{0,t}},
\end{eqnarray}
with $\val(L_e)$ expressed by Equation (\ref{eq:water_filling_min}). By defining this way the charging unit price, the aggregator makes EV and nonflexible appliances pay equally (per energy unit) for the total electricity costs caused by their aggregated electricity consumption.
This means that for consumers, electricity usages during peak hours are as much expensive as during off-peak hours.
This also means that households may have a smaller electricity bill thanks to the efforts made by the EV community.

This function $\lambda_e$ is differentiable for any $L_e $, even though function $\val$ was piecewise-defined in Prop.~\ref{prop:water_filling}.
Additionally, the following section will show the importance of the interval of $L_e$~values~for~which~$\lambda_e$~is~increasing.~To~this~end, the next Proposition gives a necessary and sufficient condition -- on nonflexible load $\bm{\ell}_0$ -- to have an increasing $\lambda_e$.

\begin{proposition}
The charging unit price function $\lambda_e$ is increasing on $\mathbb{R}_+^*$ if and only if:
\begin{equation}
\frac{\sum_{t=1}^T \tilde{\eta}_t {\tilde{\ell}_{0,t}}^n}{\sum_{t=1}^T \tilde{\ell}_{0,t}} \quad \leq  \quad n\,,
\quad \text{with }
\begin{cases}
\tilde{\ell}_{0,t} = \frac{\ell_{0,t}}{\el}\,,\\
\tilde{\eta}_t = \frac{\eta_t}{\eta_1}\,.
\end{cases}
\label{eq:suff_cond}
    \end{equation}
\label{prop:sufficient_condition}
\end{proposition}
\vspace{-0.3 cm}

\begin{proof}
\vspace{-3 mm}
\label{append:suff_cond}
The proof consists in showing that $\lambda_e'$ is increasing on $\mathbb{R}_+$ and in particular strictly increasing on $[0,L_e^{(T-1)}[$. Indeed, in that case, $\lambda_e$ would be increasing on $\mathbb{R}_+^*$ if and only if $\lambda_e'(0)\geq 0$, condition leading to Eq.~(13) of Prop.~2.
Note that the case with no nonflexible consumption ($\bm{\ell_0}=0$) is ruled out, as $\lambda_e(L_e) = L_e^{n-1}$ would be unconditionally increasing.

Even though the numerator $\val$ of $\lambda_e$ is piecewise-defined, this function is $C^1$ on $\mathbb{R}_+$ as it is based on continuous function over intervals (see Eq.~\eqref{eq:water_filling_min}). Moreover, $\lambda_e'$ is piecewise differentiable and for $L_e\in\left]L_e^{(t-1)},L_e^{(t)}\right[$ ($1\leq t\leq T$):
\vspace{-0.2 cm}
\begin{equation*}
\begin{aligned}
\lambda_e''(L_e) = \frac{2\beta_t}{(L_e+\alpha_T)^3} + H_t\frac{(L_e+\alpha_t)^{n-1}}{(L_e+\alpha_T)^3}\bigg[(n-1)(n-2) L_e^2 \\+ 2(n-2)\left(n\alpha_T-\alpha_t\right)  L_e+ \left(2\varepsilon^2 -2n\varepsilon +n^2-n\right)\alpha_T^2 \bigg] \,,
\end{aligned}
\end{equation*}
where $\varepsilon = \frac{\alpha_t}{\alpha_T}$.
The second term is non-negative since the three polynomial coefficients are non-negative for $n\geq 2$.
The first term is also non-negative, and even positive for $t\leq T-1$.
The proof is completed using the continuity of $\lambda_e'$.
\end{proof}

Note that when the marginal electricity cost (w.r.t. the aggregated charging need $L_e$: $\val'(L_e)$) is smaller than the charging unit price $\lambda_e(L_e)$, Prop.~\ref{prop:sufficient_condition} does not hold and $\lambda_e$ is (locally) decreasing. 
Also observe that a smooth nonflexible profile ($\forall t, \, \left|\ell_{0,t}/\ell_{0,1}-1\right|$ small) leads to a low ratio in~\eqref{eq:suff_cond}, which induces an increasing $\lambda_e$.

\section{Optimal routing considering energy cost}
\label{sec:optimal}
Going back to the study of the routing game in which players' costs are defined in~\eqref{eq:gen_cost}, a natural concept of equilibrium is the  Wardrop Equilibrium (WE) defined as follows:

\begin{definition}
A flow $\mathbf{x}^*$ is a Wardrop Equilibrium (WE) if and only if:
\begin{equation}
\forall s \in \{e,g\}\,, \quad c_{r,s}(\mathbf{f}^*)\leq c_{r',s}(\mathbf{f}^*),
\label{eq:WE}
\end{equation}
for all paths $r, r'$ with $r$ such that $f_{r,s}>0$.
\label{def:WE}
\end{definition}

Literally it means that at a WE, for any given vehicle class, the travel costs on all the paths actually used ($\{r: f_{r,s}>0\}$) are equal, and less than those which would be experienced on any unused path ($\{r': f_{r',s}=0\}$).
Note that WE flow $\mathbf{f}^*$) (or equivalently $\mathbf{x}^*$) will be denoted with an asterisk thereafter.


Even if the considered routing game is nonseparable and with multiple types of flows, the structure of the charging unit price $\lambda_e$ enables us to apply a similar approach to the method of Beckmann \textit{et al.} in~\cite{BECKMANN56} which is a well-known convex optimization technique~in~order~to~find~equilibria~in~routing~games. 

\begin{proposition}
The local minima of the following constrained optimization problem are WE:
\begin{equation}
\begin{array}{c}
\displaystyle \min_{\bm{x}} \, \mathcal{B}(\bm{x})\hspace{2 mm}\mr{s.t.} \hspace{2mm}\eqref{eq:contraintes}\,, \hspace{1mm} \mr{with} \\
\displaystyle \hspace{-2 mm}\mathcal{B}(\bm{x})=\tau \sum_{a\in\mathcal{A}} \int_0^{x_a}\hspace{-1 mm} d_a(x) \mr{d}x \hspace{1 mm} + \hspace{-0.2 cm} \sum_{\substack{a\in\mathcal{A}\\s\in\{e,g\}}}\hspace{-2 mm} t_{a,s} x_{a,s} + \hspace{-0.2 cm}
\sum_{s\in\{e,g\}} \int_0^{L_s(\mathbf{x_s})} \lambda_s(x) \mr{d}x \,,
\label{eq:beckmann}
\end{array}
\end{equation}
extending the definition of the global charging need $L_e$ in~\eqref{eq:L_e} to the GV class with $L_g = m_g\sum_a x_{a,g}l_a$.
\label{prop:beckmann}
\end{proposition}

\begin{proof}
\label{append:beck}
Let a path flow vector $\mathbf{f}$ be a solution of the minimization problem~\eqref{eq:beckmann} of Beckmann's function $\mathcal{B}$ under constraints of~(1). Then there exist constants $\lambda_{k,s}$ and $\mu_{r,s}$ (for all $k\in\mathcal{K}$, $r\in\mathcal{R}$, $s\in\{e,g\}$) such that $\mathbf{f}$ is a solution of the corresponding Karush-Kuhn-Tucker (or KKT) conditions:
\vspace{-0.1 cm}
\begin{equation*}
\begin{cases}
\displaystyle\frac{\partial \mathcal{L}}{\partial f_{r,s}} = 0\,,
~\frac{\partial \mathcal{L}}{\partial \lambda_{k,s}} = 0
\\\mu_{r,s}\times f_{r,s} = 0
\\\mu_{r,s},~f_{r,s}\geq 0
\end{cases}
\hspace{-0.2 cm}
\text{where }
\begin{cases}
\displaystyle
\mathcal{L} = \mathcal{B}(\mathbf{x}) -\sum_{s=e,g} \sum_{k\in\mathcal{K}} \lambda_{k,s} h_{k,s} \vspace{-0.1 cm}\\ 
\displaystyle~ -\sum_{s=e,g}\sum_{r\in\mathcal{R}} \mu_{r,s} f_{r,s}\\
\displaystyle h_{k,s} = \sum_{r_k\in\mathcal{R}_k} f_{r_k,s} - X_s D^k\,.
\end{cases}
\end{equation*}
Using relation~(2), differentiating $\mathcal{L}$ w.r.t. $f_{r,s}$ (with $r\in\mathcal{R}_k$) gives:
\vspace{-2 mm}
\begin{equation*}
\frac{\partial \mathcal{L}}{\partial f_{r,s}} = c_{r,s} - \lambda_{k,s}-\mu_{r,s}\,,
\end{equation*}
so that using the KKT conditions: $\forall r\in\mathcal{R}_k, \, c_{r,s}\geq\lambda_{k,s}$, with equality for $r$ such that $f_{r,s}>0$.

Hence, if for any $k\in\mathcal{K}$ there are two paths $r,r'\in\mathcal{R}_k$ with different costs $c_{r,s}<c_{r',s}$, then $f_{r',s} = 0$.
Otherwise, $f_{r',s} > 0$ and $c_{r',s}=\lambda_{k,s}\leq c_{r,s} \,$, which is contradictory.
Thus, the KKT conditions correspond exactly to the Def.~1 of a WE.
\end{proof}

The result of previous Proposition gives a simple method to find WE thanks to standard optimization algorithms which can be used to solve this convex (under some conditions) minimization problem. Moreover it leads to some theoretical properties of WE, starting with the following Proposition:
\begin{proposition}
\label{cor:exist-and-uniq-we}
There exists a WE, and a sufficient condition for uniqueness is that the charging unit price $\lambda_e$ is increasing.
\end{proposition}

\begin{proof}
\label{append:uniqueness}
The proof of this proposition is composed of two parts. First, the existence of the WE is based on the continuity property of the Beckmann's function $\mathcal{B}$ on the compact feasible set, so it has at least one minimum, which is a WE according to Prop.~3.
Second, the proof of the uniqueness involves a deeper analysis. Assuming that $\lambda_e$ is increasing, the proof consists in showing that $\mathcal{B}$ is strictly convex.
Indeed, in that case, the KKT conditions are equivalent to the minimization problem~(15) under the linear constraints of~(1).
As the KKT conditions are equivalent to Def.~1 of a WE (as seen in Appendix~C), this means that the WE correspond to the minima of $\mathcal{B}$, which is unique because of $\mathcal{B}$ strict convexity.

$\mathcal{B}$ is strictly convex if its hessian $\mathcal{H}$ ($\mathcal{B}$ is twice continuously differentiable because $\lambda_e$ is $C^1$) verifies $\bm{u}^T\mathcal{H}(\bm{f})\bm{u}>0$ for all path flows $\bm{f}$ in the interior of the feasible space $\mathcal{S}^\circ$ and all directions $\bm{u}\in\mathcal{S}$.

$\mathcal{H}$ has a block structure (for any $\bm{f}\in\mathcal{S}^\circ$):
\begin{equation*}
\mathcal{H}\left(\left(f_{r_k,k,s}\right)_{r_k,k,s}
\right) =
\left(\begin{array}{c|c}
M+D & D \\
\hline
D & D
\end{array}\right),
\end{equation*}
with matrices $D=\left(D_{r,r'}\right)_{r,r'\in\mathcal{R}}$ and $M=\left(M_{r,r'}\right)_{r,r'\in\mathcal{R}}$ defined as:
\vspace{-3 mm}
\begin{equation*}
\begin{cases}
D_{r,r'} = \sum_a\delta_{a,r}~\delta_{a,r'}~d'_a\left(x_a\right)\\
M_{r,r'} = k~l_r~ l_{r'}
\end{cases}
\text{where } ~k = {m_e}^2 ~\lambda_e'(L_e)\,.
\end{equation*}

It is sufficient to verify $\bm{u}^T\mathcal{H}(\bm{f})\bm{u}>0$ only along independent directions $\bm{u}\in\mathcal{S}$ chosen as in Fig.~\ref{fig:only}.
Without loss of generality, the proof will be illustrated only on the example of Fig.~\ref{fig:only}, where there are only one class $s$ and one O-D pair $k$ (associated with a travel demand $D_k$) which is connected by three paths $r$, $r'$ and $r''$.
In this case, there are two independent directions $\bm{u}=-f_{r,s}+f_{r',s}$ and $\bm{v}$ (see Fig.~1).

If $s =g$, then:
\begin{equation*}
\bm{u}^T \mathcal{H}(\bm{f}) \bm{u} = \sum_a (\delta_{a,r'} - \delta_{a,r})^2 d'_a(x_a)\quad\geq0\,,
\end{equation*}
since $d'_a$ are positive in the interior $\mathcal{S}^\circ$.
More precisely, as $r\neq r'$, there exists $a$ such that $\delta_{a,r} \neq \delta_{a,r'}$, so that the previous inequality is strict.
The proof is similar for direction $\bm{v}$, or for $s=e$ where:
\begin{equation*}
\bm{u}^T \mathcal{H}(\bm{f}) \bm{u} = k(l'-l)^2 + \sum_a \left(\delta_{a,r'} - \delta_{a,r}\right)^2 d'_a(x_a) \quad >0\,.
\end{equation*}
\vspace{-5mm}

\end{proof}
\begin{figure}
\centering
\includegraphics[scale=0.3]{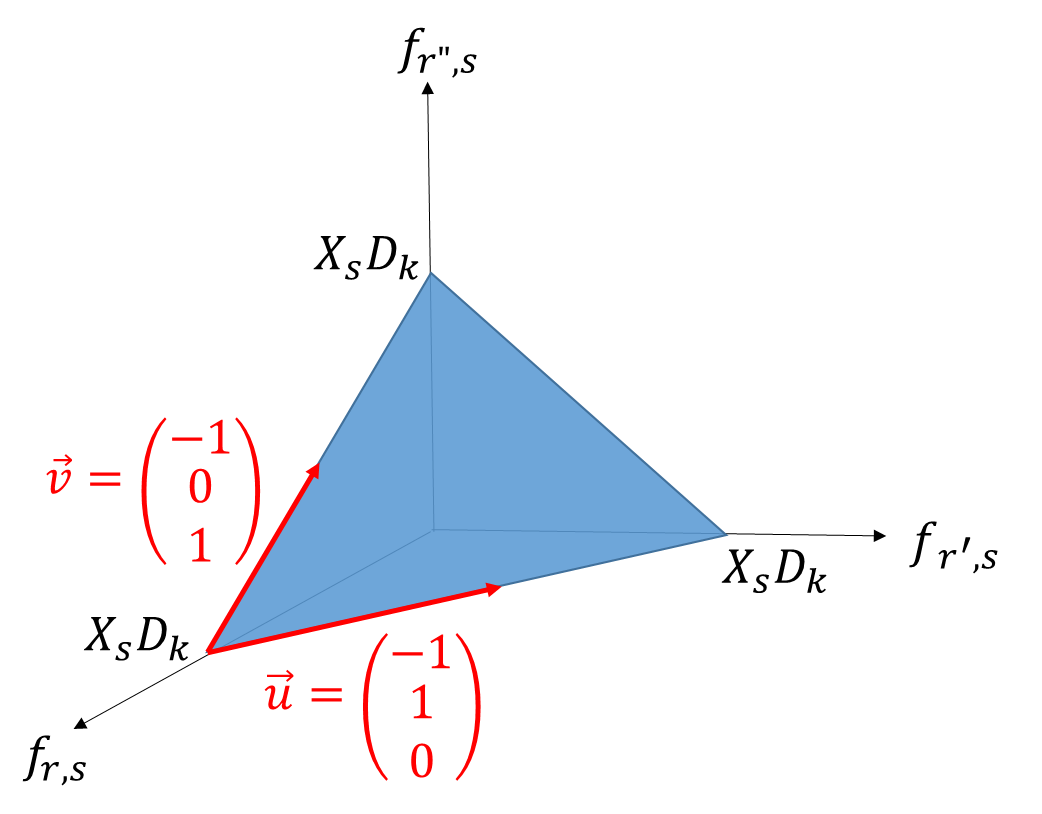}
\caption{Feasible set for class $s$ and trip $D_k$, which allows three different paths.}
\label{fig:only}
\vspace{-3 mm}
\end{figure}

In practice, the uniqueness of WE allows the TNO or the ENO to predict accurately drivers' behavior and in turn, to design good incentives according to their own objectives.
Prop.~\ref{prop:sufficient_condition} and~\ref{cor:exist-and-uniq-we} show that this property depends on the nonflexible profile.
Moreover, as any local minimum of the function defined in equation (\ref{eq:beckmann}) is a Wardorp equilibrium, standard algorithms like gradient descent or simulated annealing can be used to find these points.     
Our coupled game and this direct method to find its WE constitute a privileged framework to design incentives, compared to iterative and decomposed models~\cite{wei2017, ALIZADEH17}.
Numerical examples of incentive design are given in next section.


\section{Numerical results}
\label{sec:numerical}

The first numerical experiment will complement the theoretical analysis of the charging problem by testing on real data the validity of the equivalent condition given in Prop.~\ref{prop:sufficient_condition} for an increasing charging unit price $\lambda_e$, leading to a unique WE of the coupled driving-charging problem.
Then, the sensitivity of the WE will be studied, with respect to realistic parameters like gasoline price or when tolls are imposed in order to reduce an environmental cost. For this, experiments will be conducted on a simple framework made of three parallel roads. \\

\vspace{-2mm}

\ul{\textbf{Notations}}: because multiple days are considered in this part, $\bm{\ell_0}(d)$ is used for the nonflexible consumption~profile~of~day~$d$. \\

\vspace{-2mm}
For all the simulations, cost functions are time-independent and quadratic: $\forall t\in\{1,\dots,T\}, \, f_t(\ell_t)=f(\ell_t)=\eta \ell_t^2$.

\vspace{-0.2 cm}
\subsection{Properties of charging unit price $\lambda_e$}
\begin{figure}
\centering
\includegraphics[scale=0.4]{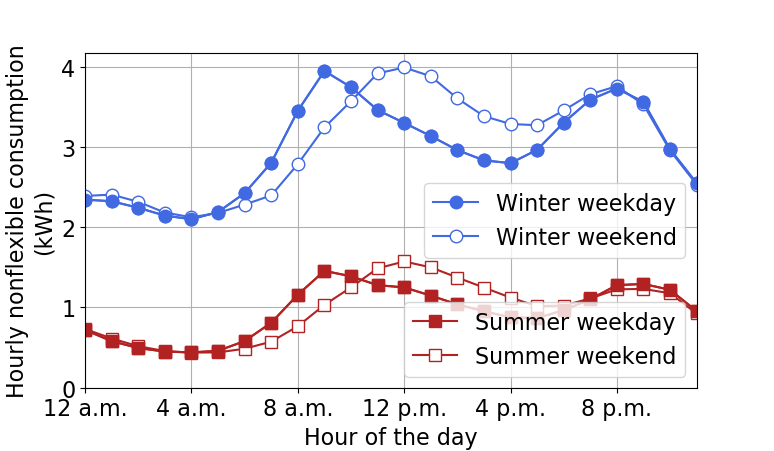}
\caption{Sample of Recoflux statistical model dataset.}
\label{fig:data0}
\end{figure}

This section focuses on two different sets of real data (resp. from France and Texas, USA) of hourly household electricity consumption throughout a year.
As discussed previously, the properties of the EV charging scheduling of Eq.~\ref{eq:pb-flex-load-schedul} depends on $\mathbf{\ell_0}(d)$ (see, e.g., Prop.~\ref{prop:sufficient_condition}). Numerically, the sensitivity of these properties to the number of time slots $T$ (directly based on their -- common -- duration) is analyzed here.

The first dataset, ``Recoflux", is from Enedis (the main French distribution network operator) and is a statistical representation of a typical French household consumption profile, taking into account electrical heating, water heating and all the other usages\footnote{Available at \url{https://www.enedis.fr/coefficients-des-profils} (RES1$\_$BASE).}.
Fig.~\ref{fig:data0} shows samples of each type~of~days in this dataset.
Note that winter consumption looks like a ``transposed version" of~the~summer~one~due~to~constant~heating.

In the raw data, daily consumption is split in 24 hourly time slots.
Fig.~\ref{fig:Pecan_Street_T}.a shows (star markers) that in this case, Prop.~\ref{prop:sufficient_condition} is not verified during whole summer, meaning that a TNO cannot rely on a unique WE to make its decisions.
To aggregate the nonflexible data to a lower number of time slots ($T<24$), these data are ordered increasingly then summed into the corresponding slots. For example for $T=2$, for each day $d$ the consumption vector $\bm\ell_{0}(d)$ is divided in two: the lower half is summed to get $\ell_{0,1}(d)$ while the higher half is summed to get $\ell_{0,2}(d)$.
Note that consecutive hourly time-slots might not be in the same new time slot, as it is the case for the off-peak hours in France, corresponding to night hours and some of the afternoon hours (from 11pm to 9am and from 3pm to 5pm.
However, the modelling choice done here leads to an underestimation of the number of days for which $\lambda_e$ is increasing.
Indeed, if the chronological order was kept, the $T$ nonflexible load values would be closer to one another (compared to the increasingly ordered ones) so that the ratio of condition~\eqref{eq:suff_cond} would be lower.

By testing the condition of Prop.~\ref{prop:sufficient_condition}, Fig.~\ref{fig:Pecan_Street_T}.a shows when the charging unit price $\lambda_e$ is increasing for each day $d$ of a given year depending on the number of time slots.
For $T>2$ time slots, summer months (from May to October) include days when $\lambda_e$ is not increasing, making it difficult for the TNO and ENO to predict drivers behavior without WE uniqueness.
This is likely to be caused by the absence of heating during these months.
In winter, nearly constant heating makes up a substantial part of consumption, lowering the impact of other electrical appliances consumption variations: the ratios $\tilde{\ell}_{0,t} = \ell_{0,t} / \el$ are then relatively low, so that Prop.~\ref{prop:sufficient_condition} is verified.
This phenomenon is worse with a higher $T$, where $\el$ is made of fewer consumption hours from raw data -- the lowest ones -- which makes it even lower, so that $\tilde{\ell}_{0,t}$ are relatively high.
In conclusion, the higher the number $T$ of time slots, the higher the proportion of days with an increasing charging unit price.

\begin{figure}
\centering
\includegraphics[scale=0.4]{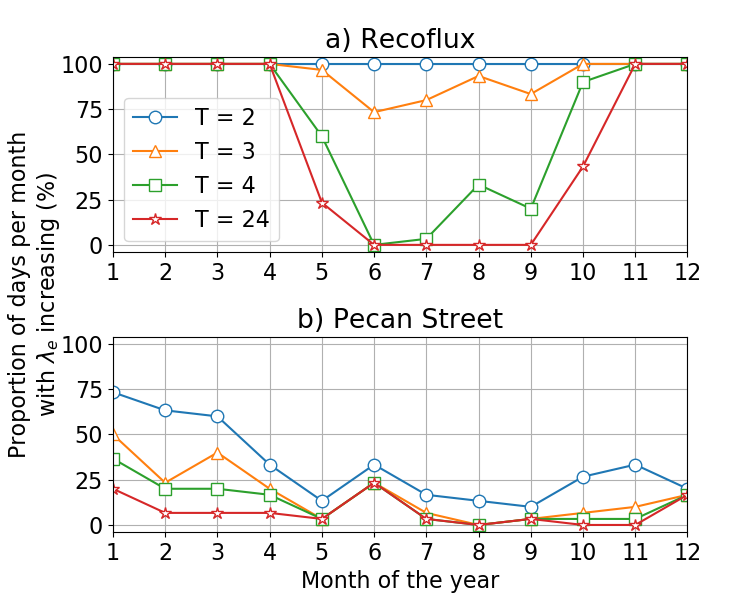}
\caption{Proportion of days per month when $\lambda_e$ is increasing (using Prop.~\ref{prop:sufficient_condition}), for different number $T$ of time slots (top: Recoflux; bottom: Pecan Street).
\textit{A lower $T$ ensures a higher proportion of days where WE is unique.}}
\label{fig:Pecan_Street_T}
\end{figure}
The second dataset is the hourly electric consumption throughout a year of a Texan household, given by the company Pecan Street\footnote{Data available at \url{http://www.pecanstreet.org/}.}.
Compared to the previous dataset, consumption here is higher during summer with the intensive use of air conditioning in Texas.
Here, $\lambda_e$ is increasing only $34\%$ of the days for $T=2$ (circle markers in Fig.~\ref{fig:Pecan_Street_T}.b).
The main reason explaining that is that Texan night consumption is very low compared to day consumption, especially in the summer period.
Note also that mean profiles like Recoflux data might smooth extreme consumption for which Prop.~\ref{prop:sufficient_condition} is not verified and $\lambda_e$ is not increasing, explaining part of the differences between the two datasets in Figs.~\ref{fig:Pecan_Street_T}.a and~\ref{fig:Pecan_Street_T}.b.


\subsection{Sensitivity of Wardrop Equilibrium}
\label{sec:param}
 \begin{figure}
     \centering
     \includegraphics[width = 0.24\textwidth]{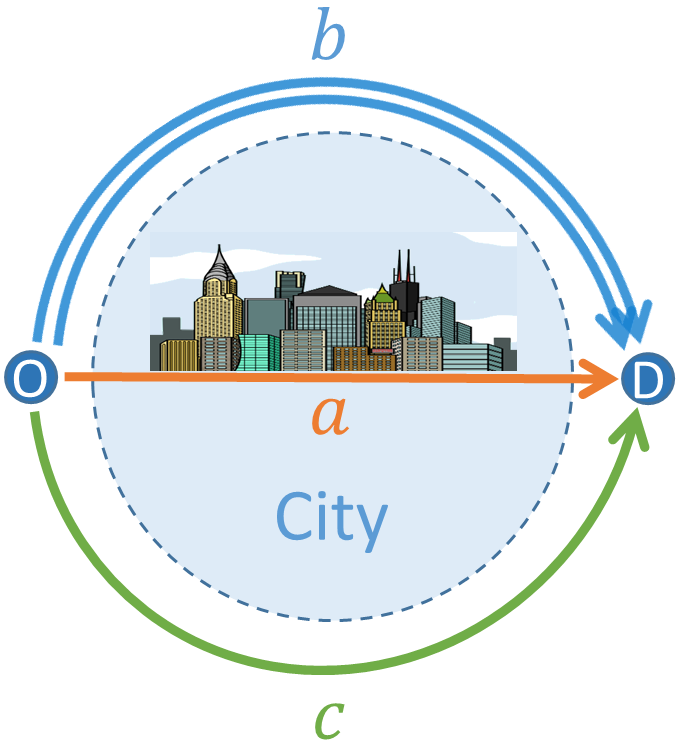}
     \caption{Schematic representation of the TAP network considered. \textit{A simple setting allowing a tractable analysis of the proposed coupled charging-and-driving concept, and an evaluation of associated urban externalities.}}
     \label{fig:schema}
     \vspace{-3 mm}
 \end{figure}

While the first numerical experiment only focused on the charging part of the problem, the two following sections consider the whole coupled problem of driving and charging, showing in particular how the WE is sensitive to the main parameters. The parameters of the problem are set as follows, unless otherwise specified: $T=2$ time slots; the nonflexible load $\bm{\ell_{0}}$ is $\ell_{0,1}= 16.7$kWh and $\ell_{0,2}= 25.6$kWh, which are the corresponding average  -- over the days of the considered year -- values in the Recoflux dataset; a quadratic electricity cost function $f_t(\ell)=f(\ell)=\eta \ell^2$ with $\eta=\num{0.01}$\euro/kWh$^2$ which allows scaling the charging unit price $\lambda_{e}$ into the interval $[0.18,0.21]$\euro/kWh for $L_e \in [0,40]$\% of total nonflexible energy consumption (order of magnitude of EV charging consumption relatively to other usages). Note that with the chosen quadratic electricity cost function, $\bm{\ell_{0}}$ verifies the condition of Prop.~\ref{prop:sufficient_condition}, which ensures the uniqueness of WE. In the simulations, WE is computed thanks to Beckmann's method as the unique solution of a strictly convex optimization problem (see explanations following Prop.~\ref{prop:beckmann}).

The road network under study here is made of two nodes, linked by three parallel links $a$, $b$ and $c$ (with $l_a<l_{b}$ and $l_a<l_{c}$). This setting represents a city where drivers can either directly cross the city (path $a$) or take any ring road (path $b$ or $c$) in order to join the opposite side of the city (thus $l_{b}=l_c=\frac{\pi}{2}l_a$).
The length $l_a = 30$km of arc $a$ is a good approximation of the daily mean individual driving distance in France (following ENTD\footnote{\textit{Enqu\^{e}te Nationale Transports et D\'{e}placements}: \url{https://utp.fr/system/files/Publications/UTP_NoteInfo1103_Enseignements_ENTD2008.pdf} (in French).} 2008).
For the speed limits on the arcs, the example of Paris was taken: $v_a=50$km/h and $v_{b} =v_c= 70$km/h.
The capacities of the arcs are set to $C_b=X$ and $C_a=C_c = X/2$, with $X$ the total number of vehicles, meaning that ring road $b$ contains twice as many lanes as the two other paths.
The parameters $\alpha = 2$ and $\beta = 4$ of the BPR functions $d_r$ defined in~\eqref{eq:BPR} are determined empirically~\cite{JEIHANI}, so that if all vehicles choose the same path, the corresponding driving time will be three times longer than the free flow reference time.
The value of time $\tau$ is set to 10\euro/h according to a French government report\footnote{\url{http://www.strategie.gouv.fr/sites/strategie.gouv.fr/files/archives/Valeur-du-temps.pdf}.}.
For GV, the consumption parameters are $\lambda_g = 1.5$\euro/L and $m_g=0.06$L/km and for EV, $m_e = 0.2$kWh/km~\cite{FONTANA13}.
Finally, 
the proportion $X_e$ of EV is supposed to be 50$\%$, in line with future predictions.
Tolls are all set to zero in order to clearly analyze the effect of congestion and energy on the behavior of EV and GV.

 \begin{figure}
    \centering
    \includegraphics[width = 0.4\textwidth]{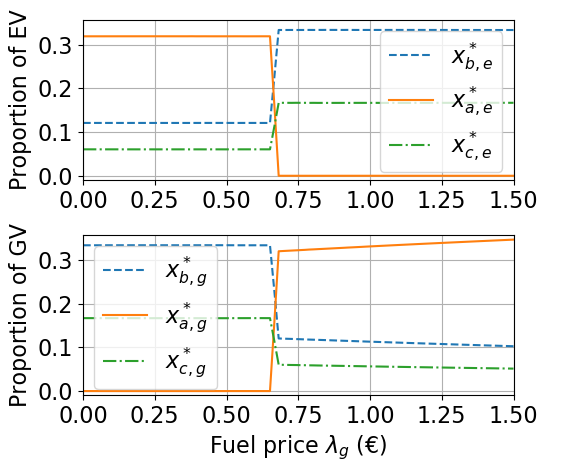}
    \caption{Evolution of the WE $\boldsymbol{x^*}$ with respect to fuel price $\lambda_g$, for $X_g=X_e=0.5$.
    \textit{There is a threshold on $\lambda_g$ corresponding to a switch of traffic equilibrium.}
    }
 \vspace{-0.4 cm}
   \label{fig:WE_lambdag}
\end{figure}
First, the sensitivity of WE w.r.t. the fuel price $\lambda_g$ is assessed in Fig.~\ref{fig:WE_lambdag}.
Starting from the actual fuel price (around $\lambda_g = 1.50$\euro/L), the cost of energy consumption per distance unit is higher for GV: $m_g \lambda_g > m_e \lambda_e(L_e(\bm{x}_e))$ for any EV flows $\bm{x}_e$.
All GV then use the shortest arc $a$ ($x_{a,g}^*=X_g=0.5$), while most of EV use the ring roads to avoid congestion ($x_{a,e}^*<X_e=0.5$).
While $\lambda_g>0.68$\euro/L the obtained WE remains the same.
Then, and down to $\lambda_g>0.65$\euro/L, the proportions of EV and GV at equilibrium are inverted.
This interval corresponds to a threshold on (exogeneous) GV energy cost when $m_g \lambda_g$ becomes lower than $m_e \lambda_e(L_e(\bm{x}_e))$ for any EV flows $\bm{x}_e$.
Then, all EV use arc $a$ while there are fewer and fewer GV on it, because going on the long ring roads is not that expensive with such a low fuel price.
In terms of decision-making, lowering taxes on fuel (leading to a smaller $\lambda_g$) may lower the level of GV traffic inside cities.
Note that the two ring roads are used simultaneously, in proportions such that there are twice as many vehicles on the ring road with the larger capacity, so that the travel time is the same for both.
\vspace{-0.2 cm}

\subsection{Minimizing Environmental Cost}
\label{sec:env}
This section focuses on the optimization of a global function of the TNO which depends on the WE $\mathbf{x^*}$. One particular function could be the level of pollution, where a TNO controls the toll prices on all paths such that the global level of pollution is minimal. Only GV release polluting substances into the air, therefore the level of pollution depends on the expected number of GV on each arc. Based on Little's formulae of queueing theory \cite{Kleinrock}, this expected number of GV on arc $a$ is the product of the rate $x_{a,g}$  and the expected travel duration $d_a(x_a)$. Note that this upper-level optimization problem is not performed by the EV aggregator but an autority/regulator of the transportation network that is the TNO. 
The TNO determines only here the GV toll $t_{a,g}$ on the arc crossing the city. For all the other traffics, there is no toll applied, which is typically the kind of incentive large urban cities like London proposed (vehicles have to pay a toll in order to go across the city downtown).

The purpose of this toll is to limit the number of GV contributing to the environmental cost, defined as:
\begin{equation}
c_\text{env}(\mathbf{x}) = \gamma_a  x_{a,g} d_a(x_a) + x_{b,g} d_b(x_b) + x_{c,g} d_c(x_c)\,,
\end{equation}
with $\gamma_a\geq 1$ the weight of environmental cost on arc $a$ (inside the city): this represents a willingness to diminish (local) pollution in the city center.
Note also that at equilibrium, the environmental cost function $c_\text{env}$ depends implicitly on the toll $t_{a,g}$ through the WE flows $\bm{x}^*$. The TNO can thus control this environmental cost, solving an upper-level optimization problem:
\begin{equation}
c_\text{env}^* = \min_{t_{a,g}\geq 0}~c_\text{env}\big(\mathbf{x^*}(t_{a,g})\big)\, ,
\label{eq:c_env}
\end{equation}
with $\mathbf{x^*}(t_{a,g})$ the (unique) WE associated to $t_{a,g}$.
As there is no explicit formulation of the WE (as a function of $t_{a,g}$), it is difficult to determine explicit solutions for this optimization problem, or even to integrate optimality conditions of the lower-level problem (WE between the vehicles) into the upper one (of the TNO). In turn, an exhaustive search on $t_{a,g}$ with a 0.01\euro~increment is performed.
\begin{figure}
\centering
\includegraphics[scale=0.4]{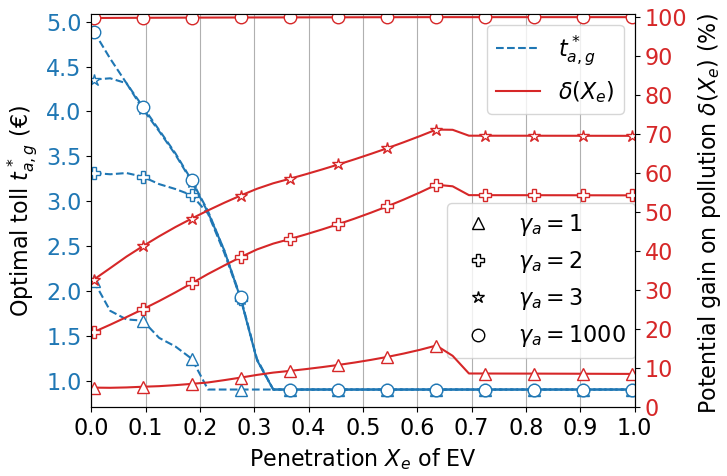}
\caption[Evolution of the optimal toll and the potential environmental gain, w.r.t. the proportion of EV.]{Evolution of the optimal toll $t_{a,g}^*$ and the potential environmental gain $\delta(X_e)$ w.r.t. the proportion $X_e$ of EV, for different $\gamma_a$ scenario.
\textit{As $X_e$ grows, the traffic operator induces a bigger impact on pollution if choosing the optimal toll, which decreases.}}
\vspace{-0.4 cm}
\label{fig:env3}
\end{figure}
\begin{figure}
\centering
\includegraphics[scale=0.45]{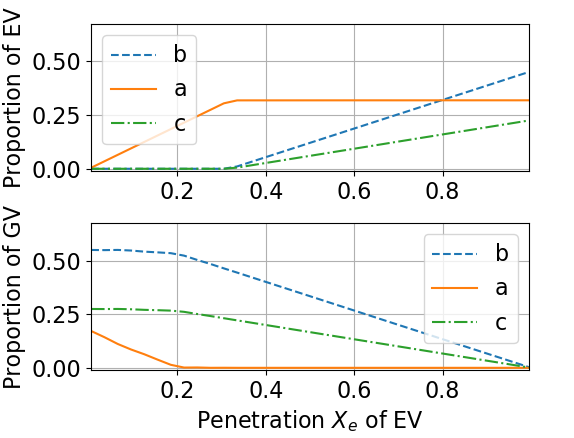}
\caption{Distribution of EV and GV at the WE which minimizes pollution, as the proportion of EV grows and for $\gamma_a = 2$. Ex: for $X_e=0.33$, few GV ($x^*_{a,g}=0.47<X_g=0.67$) while all EV ($x^*_{a,e}=X_e=0.33$) are on $a$.
\textit{As $X_e$ grows, EV replace GV on arc $a$.}}
\vspace{-0.4 cm}
\label{fig:env_eq}
\end{figure}

The proportion $X_e$ of EV has an important impact on this global minimization problem. Then, the potential environmental gain $\delta(X_e)$ by taxing optimally, relatively to the reference case $c_\text{env}^{(0)}$ obtained when $t_{a,g}=0$, depends on this proportion as:
\begin{equation} 
\delta (X_e) = \frac{|c_\text{env}^* - c_\text{env}^{(0)} |}{c_\text{env}^{(0)}}\,,~ \text{ with } c_\text{env}^* \text{ of}~\eqref{eq:c_env} \, .
\label{eq:pot_gain}
\end{equation}

For the NTO, it is interesting to know the environmental gain corresponding to an optimal determination of tolls that minimize the environmental cost. The optimal toll value $t_{a,g}^*$ depends on the proportion of EV $X_e$, as illustrated in Fig.~\ref{fig:env3}.
In this figure, $t_{a,g}^*$ and the corresponding maximal environmental gain $\delta$ are computed for any EV penetration $X_e$, for different environmental policy scenario (represented by $\gamma_a$, the importance given to local pollution inside the city).
In general, for a higher EV penetration, arc $a$ will be naturally more congested, so that a lower toll $t_{a,g}^*$ is sufficient to optimally reduce the number of GV crossing the city, which explains why the blue dotted curves $t_{a,g}^*(X_e)$ are decreasing.
Regarding the different scenario, the higher $\gamma_a$ is, the higher $t_{a,g}^*$ must be to prevent GV from crossing the city and the higher the potential environmental gain $\delta$ is.
For $\gamma_a = 1000$ (circle markers), the city's environmental cost is almost equivalent to the environmental cost only on arc $a$: In the optimal WE, there is no GV crossing the city, so that our toll mechanism corresponds to a restriction one where GV would be forbidden to cross the city.
For $\gamma_a=2$ (cross markers), there may be GV crossing the city at the optimal WE, as shown in Fig.~\ref{fig:env_eq} for EV penetrations lower than 20\%.
In this case of high GV proportion, if all GV used only the ring roads there would be too much congestion, i.e., local pollution.
For this reason, the optimal toll may be lower than the one deterring all GV from crossing the city (see the cross and circle markers in Fig.~\ref{fig:env3}).
For EV penetrations higher than 33\%, the proportion of EV crossing the city is the same (see Fig.~\ref{fig:env_eq}), so that the total cost associated with arc $a$ is constant and the optimal toll $t_{a,g}^* = 0.9$\euro~remains constant (see Fig.~\ref{fig:env3}).


\section{Conclusion}\label{sec:conc}
With the growing number of EV, new challenges related to driving and charging matters arise.
In our routing game, the coupling between driving and charging was modeled by a non-separable, non-linear charging unit price function $\lambda_e$, which depends on EV flows on the overall transportation network.
In this work, the existence of a Wardrop Equilibrium (WE) was proved for any transportation network, by generalizing Beckmann's method. 
Moreover, it was shown that the uniqueness of the WE depends on the monotonicity of the charging unit price, which depends itself on the profile of the nonflexible load. 

Numerically, the properties of WE have been analyzed on two real datasets of nonflexible loads, observing in particular that uniqueness depends on seasonal effects and electrical consumption habits. 
Other numerical experiments on a network example with parallel arcs illustrate two kinds of incentives sent by the TNO and their effects on the WE.
First, lowering taxes on fuel incite GV to use longer arcs (typically ring roads).
Second, a toll system helps to control the proportion of GV on the shorter arc (typically crossing a city center).

In a future work, an accurate model of the distribution grid will be added to our framework for a better description of the costs caused by the charging operation.
On top of that, a theoretical framework will be set to study more thoroughly the incentive problem for all the operators, making it a bi-level game.

\vspace{-3mm}


\bibliographystyle{ieeetr}

\bibliography{myrefs.bib}


\begin{IEEEbiographynophoto}{Benoît Sohet}
graduated from \'Ecole Normale Sup\'erieure de Cachan in 2018.
He received the M.Sc. degree in climate physics at Universit\'e de Paris-Saclay in 2017 and the M.Sc. degree in applied mathematics at ENSTA ParisTech in 2018.
He is currently pursuing a Ph.D. in applied mathematics at EDF R\&D and at Avignon Universit\'e.
His research is related to game theory and its applications to the electrical and transportation systems.
\end{IEEEbiographynophoto}

\begin{IEEEbiographynophoto}{Yezekael Hayel}
(M’08, SM'17) received the M.Sc. degree in computer science and applied mathematics from the University of Rennes~1 in 2002, and the Ph.D. degree in computer science from the University of Rennes~1 and INRIA in 2005. He is an Assistant/Associate Professor with the University of Avignon, France, since 2006. He has held a tenure position (HDR) since 2013. He was a Visiting
Professor with the NYU Polytechnic School of Engineering from 2014 to 2015. He was the Head of the Computer Science/Engineering Institute with the University of Avignon from 2016 to 2019. His research interests include the performance
evaluation and optimization of complex network systems based on game theoretic and queuing models. He was involved at applications in communication/ transportation and social networks, such as wireless flexible networks, bio-inspired and self-organizing networks, and economic models of the Networks. He is associate editor of the GAMES journal.
\end{IEEEbiographynophoto}

\begin{IEEEbiographynophoto}{Olivier Beaude}
received the M.Sc. degree in applied mathematics and economics from \'Ecole polytechnique, Palaiseau, France, in 2010, the M.Sc. degree in optimization, game theory, and economic modeling from Universit\'e Paris VI, Paris, France, in 2011, and the Ph.D. degree from the Laboratory of Signals and Systems, CentraleSup\'elec, and the Research and Development Center, Renault, France, in 2015. He is currently a Research Engineer with EDF Research and Development, where his research interests include controlling the impact of EV charging on the grid, and the development of coordination mechanisms, and incentives for local electricity systems.
\end{IEEEbiographynophoto}

\begin{IEEEbiographynophoto}{Alban Jeandin}
holds an Engineering degree from T\'el\'ecom Paris and \'Ecole Nationale des Ponts et Chauss\'ees, with specialization in telecommunications, intelligent transportation systems and smart cities. He joined EDF R\&D in 2011. From 2011 to 2016, he was involved in  the  French  smart  metering  roll out program. He is currently leading EDF R\&D’s research project on vehicle-grid integration and smart charging. His research interests include Smart Metering, Smart Grids, and impact of EV charging on the grid, data exchange standards for Smart Charging and V2G, and Smart Cities.
\end{IEEEbiographynophoto}

\vfill


\end{document}